\documentclass[11pt]{article}
\usepackage[T1]{fontenc}

\usepackage{times}
\usepackage{algorithm}
\usepackage{hyperref}
\usepackage{algorithmic}
\usepackage{amsmath} 
\usepackage{amsthm}
\usepackage{amssymb}
\usepackage{thmtools}
\usepackage[letterpaper, margin = 1in]{geometry}

\newtheorem{theorem}{Theorem}

\newtheorem{lemma}[theorem]{Lemma}

\newtheorem{observation}[theorem]{Observation}
\newtheorem{corollary}[theorem]{Corollary}
\usepackage{graphicx}
\usepackage{color}
\newcommand{\ignore}[1]{}
\usepackage{booktabs}
\usepackage{subcaption}
\usepackage{booktabs} 
\usepackage{array}   
\usepackage{url}
\usepackage{authblk}
\usepackage[symbol]{footmisc}

\begin{document}
\title{Effective Traveling for Metric Instances of the Traveling Thief Problem}

\author[1]{Jan Eube}
\author[2]{Kelin Luo}
\author[3]{Aneta Neumann}
\author[3]{Frank Neumann}
\author[1]{Heiko R{\"o}glin}

\affil[1]{University of Bonn, Germany}
\affil[2]{University at Buffalo, Buffalo, USA}
\affil[3]{ Adelaide University, Adelaide, Australia}

\maketitle

\begin{abstract}

The Traveling Thief Problem (TTP) is a multi-component optimization problem that captures the interplay between routing and packing decisions by combining the classical Traveling Salesperson Problem (TSP) and the Knapsack Problem (KP). The TTP has gained significant attention in the evolutionary computation literature and a wide range of approaches have been developed over the last 10 years. Judging the performance of these algorithms in particular in terms of how close the get to optimal solutions is a very challenging task as effective exact methods are not available due to the highly challenging traveling component.
In this paper, we study the tour-optimization component of TTP under a fixed packing plan. We formulate this task as a weighted variant of the TSP, where travel costs depend on the cumulative weight of collected items, and investigate how different distance metrics and cost functions affect computational complexity. We present an \(O(n^2)\)-time dynamic programming algorithm for the path metric with general cost functions, prove that the problem is NP-hard even on a star metric, and develop constant-factor approximation algorithms for star metrics. 
Finally, we also develop an approximation algorithm for the problem under a general metric with a linear cost function.

We complement our theoretical results with experimental evaluations on standard TTP instances adjusted to a path metric. Our experimental results demonstrate the practical effectiveness of our approaches by comparing it to solutions produced by popular iterative search algorithms. The results show that our methods are able to significantly improve the quality of solutions for some benchmark instances by optimizing the traveling part while pointing out the optimality of the travel component for other solutions obtained by iterative search methods.
\end{abstract}

\section{Introduction}
Many real-world optimization problems involve several NP-hard combinatorial optimization problems that need to be solved in combination to come up with high quality solutions for the complex problem at hand. Such problems being a combination of different hard combinatorial optimization problems are known as multi-component problems in the literature~\cite{DBLP:books/sp/19/BonyadiM0019}. Iterative search algorithms such as different variants of local search, evolutionary algorithms and ant colony optimization have been applied to the well-known traveling thief problem (TTP) which is a multi-component problem that combines the traveling salesperson problem (TSP) and the knapsack problem (KP). 

Understanding the interactions between these problems and how the combination of different subproblems impacts the complexity of solving them is a challenging task~\cite{DBLP:journals/ec/PrzybylekWM18} and has been mainly studied from an empirical perspective~\cite{DBLP:conf/gecco/PolyakovskiyMM17,DBLP:journals/swevo/HerringKY24}. With this paper, we contribute to this area of research from a theoretical perspective and investigate how to solve the class of weighted traveling salesman problems where travel costs increase monotonically with the weight of items collected while traveling. Our developed approaches allow to solve the traveling component of TTP in an optimal way in restricted settings. This allows to judge the quality of solutions obtained by popular iterative search algorithms for TTP and increase their theoretical undestanding in terms of performance gaps.

\subsection{Related work}
The traveling salesman problem (TSP) is one of the best studied combinatorial optimization problems. It has been investigated extensively both from a theoretical and practical perspective and a wide range of powerful TSP solvers are available.
In real world settings, routing problems such as the TSP are often combined with other problems such as packing. In this paper, we consider the class of TSP problems with such packing components influencing the travel time when moving from a node $i$ to node $j$ in the TSP. Our setting is motivated by recently investigated problems such as the traveling thief problem (TTP)~\cite{Bonyadi2013TTP} and the weighted TSP~\cite{DBLP:conf/gecco/BossekCK020} for which the travel cost when moving between cities increases monotonically with respect to the weight of items collected at previously visited cities.
The TTP combines the TSP with the knapsack problem, and a combination of a tour and a packing plan needs to be found that optimizes an objective function which accounts for the profit of collected items and non-linear travel costs that depend on the weight of the collected items while visiting the cities.
A wide range of approaches for the TTP and its variants have been proposed in the literature~\cite{DBLP:conf/gecco/PolyakovskiyB0MN14,DBLP:journals/ec/PrzybylekWM18,DBLP:journals/swevo/HerringKY24,DBLP:journals/telo/NikfarjamNN24}. 

The case of finding an optimal packing plan for a fixed tour is already NP-hard. However, mixed integer programming~\cite{DBLP:journals/eor/PolyakovskiyN17} and dynamic programming~\cite{DBLP:conf/algocloud/NeumannPSSW18} approaches are available and have been shown to solve the packing part when given a fixed tour exactly or approximately for up to $1000$ items. Finding an optimal tour for a given packing plan on the other hand seems to be much more difficult and no effective exact approaches or approximation algorithms are currently available if the weighting of the travel distances is non-linear. If the packing plan is fixed, then the problem turns into a weighted TSP problem. Linear-weighted metric TSP problems have been examined in \cite{DBLP:conf/gecco/BossekCK020} and different constant factor approximation algorithms have been introduced. In the case that each node contributes a weight of $1$, the relationship to the maximum latency problem~\cite{DBLP:conf/stoc/BlumCCPRS94} has been used to come up with constant-factor approximation algorithms and the results have also been generalized to polynomial weights.
Improved approximation results for the case where the TSP only includes distance of $1$ or $2$  have been obtained by exploring the connection to approximation algorithms for the classical $\{1,2\}$-TSP~\cite{DBLP:journals/mor/PapadimitriouY93}, which is already NP-hard as it is a generalization of the Hamiltonian cycle problem.

\subsection{Our contribution}
With this paper, we contribute to solving the tour optimization part of such weighted TSP instances and provide algorithms with approximation guarantees and polynomial time dynamic programming approaches for important weighted TSP formulations. 
We first prove that the Weighted Traveling Salesman Problem (W-TSP) remains NP-hard, establishing that the problem is computationally challenging even in highly structured settings. Despite this hardness, we design an $(8+\varepsilon)$-approximation algorithm running in polynomial time for this metric. This result demonstrates that strong approximation guarantees are achievable for nontrivial weighted variants of TSP and helps delineate the boundary between tractable and intractable cases. 
For W-TSP with a general cost function under a path metric, we develop polynomial-time dynamic programming algorithms for both fixed and arbitrary starting points, which form the core tour-optimization component of our experimental evaluation. By integrating this exact W-TSP solver into the TTP framework, we obtain substantially improved solution quality compared to existing approaches. Our experimental results show that the line-metric W-TSP algorithm consistently outperforms previously known heuristics on the considered TTP instances. To further enhance scalability, we introduce a grouping heuristic that accelerates the dynamic programming computation while preserving near-optimal solution quality, demonstrating the practical impact of our theoretical contributions and their effectiveness in advancing the state of the art for TTP. 

The paper is structured as follows. 
In Section~\ref{sec2}, we introduce the W-TSP which is the problems subject to our investigations and show in Section~\ref{sec3} that W-TSP is NP-hard even under restricted metric settings for the TSP part. Section~\ref{sec4} provides approximation and dynamic programming approaches for important settings. Our experimental investigations carried out in Section~\ref{sec5} show that the dynamic programming approach provides significant improvements when the underlying metric for the TSP is a path metric. Finally, we finish with some concluding remarks.
 
\section{Problem formulation} 
\label{sec2}
We introduce and study the \emph{Weighted Traveling Salesman Problem} (W-TSP), which generalizes the classical TSP, and also provides a permutation-based solution to the \emph{Traveling Thief Problem} (TTP) when the packing plan is fixed.  
The input is given as a TSP instance with cities $v_1, \ldots, v_n$ and distances $d(v_i, v_j)$ between them, and a set of $n$ items with weights located at the different cities. We denote by \( w_{v_i} \) the weight of the item located at node \( v_i \), and let \( f(w) \) denote the cost per unit distance of traversing while carrying total weight \( w \). To simplify writing we sometimes write $w_i$ instead of $w_{v_i}$.
The goal is to compute a TSP tour $\pi$ that minimize the overall cost. 
We denote by $\pi = (\pi_1, \ldots, \pi_n)$ the permutation of the cities which determines the order in which the cities are visited. This presents a solution to the TSP part of the TTP problem.  

In the classical TTP~\cite{Bonyadi2013TTP}, the function $f$ is a monotone function given by the inverse of the speed of the vehicle. The speed decreases linearly with respect to the weight of the items already collected. Furthermore, each tour starts at city \( v_1 \), and either no item is located at \( v_1 \), or the item located at \( v_1 \) is collected at the end of the tour. 

We are interested in general monotone functions $f$ that weight the distances which need to be traveled. 
Our goal is to find a permutation (also referred to as a tour) $\pi=(\pi_1, \pi_2, ..., \pi_n)$ which minimizes 
the travel cost given as 

$$
 T(\pi) =  f\left(\sum_{j=2}^n w_{\pi_{j}}\right) \cdot d(\pi_{n}, \pi_{1})   +  \sum_{i=1}^{n-1} f\left(\sum_{j=2}^i w_{\pi_{j}}\right) \cdot d(\pi_i, \pi_{i+1}).
$$

Related to this problem is the node weighted TSP problem introduced in \cite{DBLP:conf/gecco/BossekCK020}. Here, each node has a given weight and distances are weighted by the sum of the weights of the nodes that have already been visited except the first node. It is a special case of our formulation when the function of $f$ returns the sum of the weight of the nodes visited until city $\pi_j$. In contrast, our work considers a more general cost function \( f \), which does not necessarily satisfy the pure accumulated-weight property exploited in \cite{DBLP:conf/gecco/BossekCK020}. As a result, several algorithmic techniques and constant-factor approximation strategies developed for the node weighted TSP do not directly extend to our setting, and our approximation results apply to a strictly broader class of weighted TSP objectives.
 
 \section{NP-Hardness}  
 \label{sec3}
Note that when \( f(w) = 1 \), the W-TSP reduces to the classical Traveling Salesman Problem (TSP), which is known to be APX-hard. 
This section establishes that the W-TSP problem remains NP-hard even under restricted metric settings. Specifically, we show that W-TSP is NP-hard when the underlying distance metric is a special star metric.

We consider a special star metric, where a star represents a tree of height 1 with a root denoted as $v_0$ and leaves as $\{v_1, ..., v_n\}$. We define the metric on these nodes as follows: for all $i \in [n]$, the distance between $v_0$ and $v_i$ is denoted as $d(v_0,v_i)$, and for all $i, j \in [n]$ where $v_i\neq v_j$, the distance between $v_i$ and $v_j$ is defined as $d(v_i,v_j) = d(v_0, v_i) + d(v_0, v_j)$. 

 \begin{theorem} Even when the metric is a star metric, the W-TSP problem  is NP hard.\end{theorem}    
 
\begin{proof}
We reduce the partition problem to the W-TSP with a star metric. For this, consider an instance
of the partition problem, i.e., given is a set $S$ of positive integers $\{s_1, s_2, ..., s_n\}$, and let $\lambda= \sum_{i\in [n]} s_i$,  
  The question is if $S$ can be partitioned into two subsets $S_1$ and $S_2$ such that the sum of the numbers in $S_1$ equals the sum of the numbers in $S_2$.    This problem is NP-hard~\cite{DBLP:books/fm/GareyJ79}.

For the given  instance of the partition problem, we create an instance of the W-TSP  with a star metric as follows. The instance is defined on the star metric $(V, E)$ with a root $v_0$ and leaves  $ \{v_1, ..., v_n, v_{n+1}\}$, $ d(v_0, v_i)= s_i  $ for all $i\in [n]$ and $d(v_0, v_{n+1})=s_{\max} := \max_{i\in [n]} s_i$.  Each node $v_i$ has  weight $w_i=s_i$ for all $i\in [n]$,  node $v_{n+1}$ has weight $w_{n+1}=s_{\max}$ and  node $v_0$ has weight $w_0=\lambda + s_{\max} + 1$. Finally, we specify the  cost function $f(w)$ as a simple piecewise function as follows:
  \begin{equation}
f(w)=
    \begin{cases}
        0 & \text{if } w \le \lambda/2\\
        1 & \text{if } \lambda/2 < w \le   \lambda+ s_{\max} \\
        \infty & \text{if } w >  \lambda+ s_{\max} 
    \end{cases}
\end{equation}  
Then the problem is to find a minimum W-TSP tour $\pi=(\pi_1,\ldots,\pi_{n+2})$  of $V$  with minimum cost, i.e., \[d(\pi_{n+2}, \pi_1) \cdot f\Big(\sum_{j=2}^{n+2} w_{\pi_j}\Big)  + \sum_{i=1}^{n+1} \Big(d(\pi_i, \pi_{i+1}) \cdot f\Big(\sum_{j=2}^i w_{\pi_j}\Big)\Big).\]   

We show that there is a valid partition $S_1, S_2$ iff the W-TSP instance has a solution of cost at most $ s_{\max}+ \lambda  $.

Let $S_1, S_2$ be a valid partition such that the sums of the elements in $S_1$ and $S_2$, respectively, are both $\lambda/2$. Then the following tour provides a solution for the W-TSP with travel cost $ s_{\max}+ \lambda$: 
starting from node $v_0$ (and pick up items in node $v_0$ at the end), first visit  all nodes corresponding to the integer set $S_1$, which costs $0$ as $f(w) =0$ for $w\le \lambda/2$. Then visit $v_{n+1}$, which increases the collected weight to $\lambda/2+s_{\max}$. Then visit all nodes corresponding to the set $S_2$ and return to $v_0$ afterwards. The total length of the subtour starting at $v_{n+1}$ and visiting all nodes corresponding to $S_2$ and returning to $v_0$ is $s_{\max}+2\sum_{s_i\in S_2}s_i=s_{\max}+\lambda$. Since $f(w)=1$ for $\lambda/2 < w \le   \lambda+ s_{\max}$, the total contribution of this subtour to the costs is $s_{\max}+\lambda$.
Finally, pick up node \( v_0 \), which incurs no additional cost.

For the other direction, assume we have a tour $\pi$ which provides a solution with travel cost no more than $ s_{\max}+ \lambda $   
for the W-TSP. 
Observe that the root \( v_0 \) must be the start node of the tour and is also the last node whose weight is collected. Otherwise, the tour would require returning to the start node after visiting \( v_0 \), during which the accumulated weight would include the additional weight \( w_0 = \lambda + s_{\max} + 1 \), resulting in \( f(w) = \infty \) and hence an infinite travel cost.
Therefore, without loss of generality, we assume that the tour starts at \( v_0 \). 
For the remaining nodes, without loss of generality, suppose node $v_h$ with $h\in [n+1]$ is the latest visited node in $\pi$ such that the total weight of all previous visited nodes is no more than $  \lambda/2$, and the total weight of the previous visited nodes is $X$. Notice that $X\le \lambda/2$ and $X+w_h> \lambda/2$. There are two cases: (1) $w_h=s_{\max}$. 
Then the total cost of the tour \( \pi \) is \( s_{\max} + 2(\lambda - X) \), which is at most \( s_{\max} + \lambda \) if and only if the partition instance admits a valid solution. Otherwise, \( X < \lambda / 2 \), and thus \( \lambda - X > \lambda / 2 \) implies \( s_{\max}+ 2(\lambda - X) > s_{\max} + \lambda\). 
(2) $w_h \neq s_{\max}$. Then the total cost of tour $\pi$ is $ s_h+ 2(\lambda +s_{\max} - X - s_h)= \lambda +s_{\max} + ( \lambda -2 X) + (s_{\max}-s_h)  > \lambda +s_{\max}$ since  $( \lambda -2 X)\ge 0$ and $s_{\max}-s_h> 0$. 
Thus, the constructed W-TSP instance admits a solution with total cost at most \( s_{\max} + \lambda \) if and only if the original partition instance has a valid partition. 
This establishes the correctness of the reduction and completes the NP-hardness proof.
\end{proof}

\section{Algorithms}   
\label{sec4}
In this section, we present three algorithms for constructing a visiting tour (also referred to as a permutation of all cities), each tailored to a different setting: two designed for special metrics—the star metric and the path metric—and one for a setting with a specialized linear cost function.

\subsection{Star metric: A constant approximation algorithm}
 First, we simplify the problem by restricting attention to tours that start and end at the center node, with the center node’s weight collected at the end. This restriction is without loss of generality. Indeed, consider any feasible tour \( \pi \) of the original problem that starts at a leaf node \( v \) and ends at \( v \), collecting the weight of \( v \) at the end. We construct a corresponding tour \( \pi' \) that starts and ends at the center node by prepending the edge \( (v,\text{center}) \) to the beginning of \( \pi \) and appending the edge \( (\text{center},v) \) to its end.

The additional cost incurred by this transformation is a fixed quantity determined solely by the distance between \( v \) and the center and the total accumulated weight at the end of the tour; in particular, it is independent of the order in which the remaining nodes are visited. Conversely, any tour starting and ending at the center can be transformed into a tour starting and ending at an arbitrary leaf node \( v \) by removing these two fixed segments, adjusting the total cost by the same constant. 
Therefore, optimizing over tours that start and end at the center is equivalent to optimizing over tours with arbitrary starting nodes, up to an additive constant. As a result, we may, without loss of generality, focus on the simplified setting in which the tour starts and ends at the center node. In this setting, visiting any other node consists of a trip from the center to that node followed by a return trip back to the center with the accumulated weight.

\begin{observation}
Any optimal tour will have the same total distance, equal to twice the sum of all edge lengths in the star graph;
Any tour will start with same the same weight $0$ and have the same total weight at the end, equal to the sum of the weights of all nodes. 
\end{observation}

Given a set of leaf nodes \( J \), let \( D(J) \) denote the total distance required to traverse all nodes in \( J \), starting and ending at the center. That is,
\[
D(J) = \sum_{v \in J} (d_{\rightarrow v} + d_{\leftarrow v}),
\]
where \( d_{\rightarrow v} \) is the distance from the center to node \( v \), and \( d_{\leftarrow v} \) is the distance from node \( v \) back to the center.
 
Our algorithm follows the same structure as a scheduling algorithm developed by Epstein et al.\ which aims at obtaining a sequence of jobs that is robust against unforeseen changes in the speed of a given machine \cite{epstein2010universal}. We end up with the following algorithm:
 
\begin{enumerate}
\item Scale the distances so that the minimum distance \( d_{v} \geq 1 \) for every \( j \in J \).
    \item For each \( i \in \{1, \ldots, \lceil  \log D(J) \rceil\} \), find a subset \( J^*_i \) of leaf nodes such that the total distance \( D(J^*_i) \leq 2^i \) and the total weight \( w(J^*_i) \) is maximized.\\
    \emph{Note:} \( J^*_{\lceil  \log D(J) \rceil} = J \).
    
    \item Construct a tour \( \pi \) as follows:
    \begin{itemize}
        \item Initialize an empty tour.
        \item For \( i = \lceil  \log D(J) \rceil\) down to 1:
        \begin{itemize}
            \item Append the nodes in \( J^*_i \setminus \bigcup_{k=0}^{i-1} J^*_k \) in any order to the end of the tour.
        \end{itemize}
    \end{itemize}
\end{enumerate}

 For a tour \( \pi \), we decompose the tour into \emph{round-trip segments}, where each segment consists of traveling from the center node to a leaf node and then returning from that leaf node back to the center. During each such round trip, the vehicle carries an accumulated weight corresponding to the total weight of all nodes collected before completing the return to the center.

For any weight threshold \( w \), let \( D^{\pi}(w) \) denote the total length of all round-trip segments in tour \( \pi \) whose return trips are completed while carrying an accumulated weight of at least \( w \). Equivalently, \( D^{\pi}(w) \) is the sum of the lengths of all center-to-leaf-to-center segments for which the accumulated weight on the return leg is at least \( w \). Figure~\ref{fig:startalg} illustrates this definition.

\begin{figure}[t]
\centering
\includegraphics[width=7.5cm]{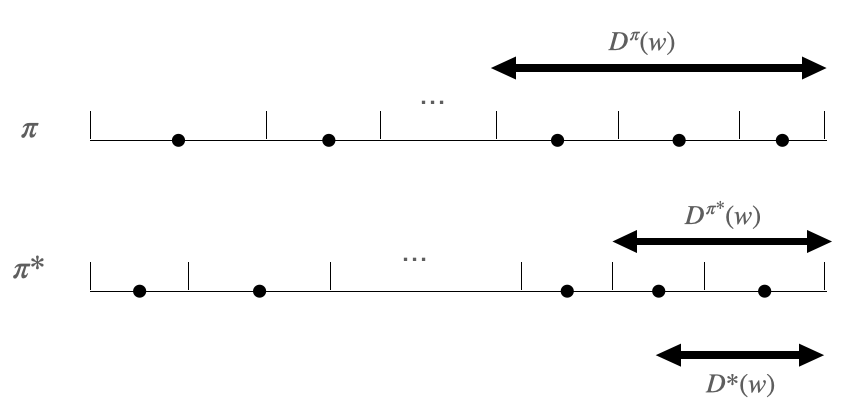}
\caption{Round-trip segment length with bounded weight.}
    \label{fig:startalg}
\end{figure}

Similar to the scheduling problem, we can prove that: 
\begin{lemma}Consider an optimal tour \( \pi^* \) and a given tour $\pi$.    
If $D^{\pi}(w) \leq c \cdot D^{\pi^*}(w)$ holds for any $w$, 
then the total cost of tour \( \pi \) is at most \( 2 \cdot c \) times the total cost of the optimal tour.
\end{lemma}
\begin{proof}
Let \( J' \subset J \) be a subset of leaf nodes visited in the first part of tour \( \pi \), and let \( w \) denote the total weight of \( J' \) together with the next leaf node visited in \( \pi \).  
Define \( D^*(w) \) to be the minimum total length traversed by the optimal tour \( \pi^* \) with accumulated total weight of at least \( w \). 
Let \( v \) be the first leaf node in \( \pi^* \) such that the accumulated weight reaches or exceeds \( w \). By definition, we have:
\[
D^*(w) = D^{\pi^*}(w) - d_v \geq \frac{1}{2} D^{\pi^*}(w),
\] 
since \( d_v \leq D^{\pi^*}(w) - D^*(w) \) and the minimum distance to a node is at least half the total length. 

Given that \( D^{\pi}(w) \leq c \cdot D^{\pi^*}(w) \), it follows that:
\[
D^{\pi}(w) \leq 2c \cdot D^*(w).
\]

By continuing with the the same steps as in the analysis by Epstein et al.~\cite{epstein2010universal} one then obtains that the total cost of \( \pi \) is at most \( 2c \) times that of \( \pi^* \). 
\end{proof}

\begin{lemma} 
\label{lem:fourtimesdistance}
If the metric is a star metric, our algorithm produces a tour $\pi$ such that 
\[D^\pi(w) < 4 D^{\pi^*}(w)\] for all $w \geq 0$.
\end{lemma}

\begin{proof}    
Let $0<  w < w(J)$ and let $i$ be minimal such that $w(J^*_i) \geq w(J) - w$.  
By construction of $\pi$, leaf nodes $j \in \bigcup_{k=1}^i J^*_k$ have an accumulated weight greater than $  w$ when they are visited. Thus,
\[
D^\pi(w) < \sum_{k=1}^{i} D(J^*_k) \leq \sum_{k=1}^{i} 2^k = 2^{i+1} - 2. \tag{1}
\]  
When \( 0 < w < w(J) \), it holds that \( D^*(w) \geq 2 \), because we scale the distances so that the minimum distance \( d_{\rightarrow v} \geq 1 \) for every \( j \in J \). 

In case $i = 1$, the claim is trivially true, and thus $D^*(w) \ge D^\pi(w) =2$. 

Suppose $i \geq 1$, then by our choice of $i$, it holds that $w(J^*_{i-1}) < w(J) - w$. Therefore, in any sequence $\pi'$, the total distance (start and end with center) 
during which the accumulated weight on the return trip to the center is at least \( w \) is larger than $2^{i-1}$, because otherwise we get a contradiction to the maximality of $w(J^*_{i-1})$.

That is, $D^*(w) > 2^{i-1}$. Together with (1), this concludes the proof.
\end{proof} 

By combining the previous two lemma one obtains:

\begin{theorem}
    Our algorithm computes an $8$-approximation for the weighted TSP problem on star graphs.
\end{theorem}

However, in Lemma~\ref{lem:fourtimesdistance}, we implicitly ignored the computational complexity of the step asserting that the total distance (starting and ending at the center) for which the accumulated weight on the return trip to the center is at least $w$ exceeds $2^{i-1}$. In fact, Step~(2) of the algorithm relies on a subroutine that solves a knapsack-type problem, namely maximizing the total collected weight of a subtour subject to a length constraint. This problem is NP-hard and admits only a $(1+\varepsilon)$-approximation in polynomial time \cite{IbarraK75}.

Therefore, when restricting attention to polynomial-time algorithms, the above result should be stated as follows.

\begin{theorem}
 There exists an algorithm that computes an $(8+\varepsilon)$-approximation in polynomial time for the weighted TSP problem on star graphs.
\end{theorem}

\subsection{Path metric: DP algorithm} 

 In this section, we consider a path metric with set of nodes $\{v_1, v_2, ... \}$ ordered from left to right. We define the metric on these nodes by $d(v_i, v_{i+1})$ for all $i\in [n-1]$, and $d(v_i, v_j)=d(v_j, v_i)=\sum_{h=i}^{j-1} d(v_h, v_{h+1}) $ for all $i<j$.   

We begin by solving the problem under the assumption that the starting node is fixed, and its weight is collected at the end of the traversal.
Let \( t \) denote the starting node, where \( t \in \{ v_1, v_2, \dots, v_n \}  \).
Our problem involves \( n+1 \) nodes, \( \{ v_1, v_2, \dots, v_n \} \cup \{ t \} \), where \( \{ v_1, v_2, \dots, v_n \} \) are ordered from left to right. We will present a dynamic programming algorithm that computes an optimal solution in polynomial time. First we may observe that it does not make any sense to visit a node (except the starting node) when it is surrounded by two unvisited nodes because the agent will also pass by this node at a later point in time:

\begin{lemma}
Given a set of uncollected nodes arranged on a path, there exists an optimal tour (i.e., a node permutation) such that for all \( i < j < k \), the tour visits either \( v_i \) or \( v_k \) before visiting \( v_j \).
\end{lemma}

\begin{proof}
    For every $i <j < k$ we say that $v_j$ is visited prematurely if it is visited before $v_i$ and $v_k$. Let $\pi^*$ be an optimal tour minimizing the number of premature visits. Let us assume that there still exists a node $v_j$ that is visited prematurely by $\pi^*$. Let us assume w.l.o.g.\ that the last node visited by $\pi^*$ lies to the left of $v_j$. Let $v_k$ be the last node to the right of $v_j$ that is visited by $\pi^*$ and let $v_i$ be the node visited after $v_k$. Then we know that $v_i$ lies to the left of $v_j$. Additionally, $v_k$ is visited after $v_j$ because $v_j$ was visited prematurely. 
    
    Let us consider the tour $\pi'$ where we remove the initial visit of $v_j$ from $\pi^*$ and visit $v_j$ after $v_k$ instead. One might note that $\pi'$ is not longer than $\pi^*$ because $d(v_k,v_j) + d(v_j,v_i) = d(v_k,v_i)$. Additionally, the segments of $\pi^*$ and $\pi '$ before $\pi^*$ visits $v_j$ and the segments after $\pi'$ visits $v_j$ are exactly the same. During the segment in between the visits of $v_j$ by the two tours the only difference is that on tour $\pi^*$ the agent has already picked up the weight of $v_j$ while on $\pi '$ this is not the case which means that the agent moves at least as fast on $\pi'$ as on $\pi^*$. As a result we have that $T(\pi') \leq T(\pi^*)$. Since $\pi'$ does not visit $v_j$ prematurely anymore this contradicts that $\pi^*$ minimizes the number of premature visits among all optimal tours.
\end{proof}

One might note that any tour that visits no node that is surrounded by two unvisited nodes can at any point in time only visit either the leftmost or the rightmost node. Thus:

\begin{corollary}\label{cor:leftorrightfirst}
On a path metric there exists an optimal tour, where at any point in time, the next node selected from the remaining unvisited (i.e., uncollected) nodes on the path is either the leftmost or the rightmost among the remaining unvisited nodes.
\end{corollary}

This structural property forms the foundation of our dynamic programming algorithm. The core idea is to solve subproblems on subpaths and then combine their solutions. 
For each subpath, we consider a set of uncollected nodes within the interval \( [a, b] \subseteq [v_1, v_n]\), where \( a \) and \( b \) denote the leftmost and rightmost nodes, respectively. Let \( s \) be the starting node of the current traversal, and let \( t \) be the ending node, whose weight will be collected at the end of the tour.
For this subproblem, the total weight collected prior to visiting the nodes in the interval \( [a, b] \) is equal to the sum of the weights of all nodes located outside this interval, excluding the end node \( t \) (whose weight is collected at the end).  
Based on the lemma, we know that the first node collected in any subpath must be either \( a \) or \( b \), which allows us to recursively build the optimal tour by evaluating both options and choosing the one with minimum cost. 

In the following, we formally define the subproblems and apply dynamic programming to solve the problem over the entire path.

 Let \( g(v_i, v_j, s, t) \) represent the optimal value of the subproblem defined on the subpath \( [v_i, v_j] \), under the condition that all nodes except those on the subpath \( [v_i, v_j] \) have already been visited, which means that the respective weights have been picked up. Here, \( s \in \{ v_{i}, v_{j} \} \) denotes the currently collected node, and \( t \) is the ending node, which will be collected last. The function is computed for every pair of nodes \( v_i \) and \( v_j \) where \( v_i, v_j \in \{ v_1, v_2, \dots, v_n \} \) are the nodes on the path from \( v_1 \) to \( v_n \).  

There are $\sum^n_{i=1} \sum_{j=i}^{n} 2=n^2 + n $ 
such relevant values, in which 
for each interval of $[x_i, x_j]$ with $i\le j$, there are $2$ choices for the starting nodes.

When \( v_i = v_j \), the subpath consists of a single node. In this case, the tour starts at \( v_i \), collects its weight, and then returns to \( t \). In this case, $
g(v_i, v_j, s, t) = d(v_i, t) \cdot f(w),
$
where \( w = \sum_{i=1}^n w_i \) is the total collected weight.

We are now ready to introduce the recurrence relation.   
Consider the subproblem \( g(v_i, v_j, s, t) \), defined on the interval \( [v_i, v_j] \) with $i<j$, where the remaining uncollected nodes lie within the interval, the traversal starts at node \( s \in \{v_i, v_j\} \), and ends at node \( t \). 
The collected weights before picking nodes in interval $[v_i, v_j]$ is $w=\sum_{\ell \in [1, i)} w_\ell + \sum_{\ell \in (j, n]} w_\ell $.  
To compute \( g(v_i, v_j, s, t) \), we consider two possible cases: 
\begin{itemize}
\item  If \( s = v_i \) and the next visited node is \( v_{i+1} \), then the cost consists of the cost to visit \( v_{i+1} \) from \( s \) and the optimal cost of collecting all remaining nodes in the interval \( [v_{i+1}, v_j] \). Thus, we have:
\[
g(v_i, v_j, s, t) = f(w + w_i) \cdot d(v_i,v_{i+1}) + g(v_{i+1}, v_j, v_{i+1}, t).
\]
Similarly,  if \( s = v_j \) and the next visited node is \( v_{j-1} \), 
\[
g(v_i, v_j, s, t) = f(w + w_j) \cdot d(v_j,v_{j-1}) + g(v_{i}, v_{j-1}, v_{j-1}, t).
\]

\item  If \( s = v_i \) and the next visited node is \( v_{j} \), then the cost consists of the cost to visit \( v_{j} \) from \( s \) and the optimal cost of collecting all remaining nodes in the interval \( [v_{i+1}, v_{j}] \). Thus, we have:
\[
g(v_i, v_j, s, t) = f(w+ w_i) \cdot d(v_i, v_j) + g(v_{i+1}, v_j, v_{j}, t).
\]
Similarly,  if \( s = v_j \) and the next visited node is \( v_{i} \), 
\[
g(v_i, v_j, s, t) = f(w+ w_j) \cdot d(v_j, v_i) + g(v_i, v_{j-1}, v_i, t).
\]
\end{itemize}
 
According to Corollary~\ref{cor:leftorrightfirst}, the optimal tour for the subproblem \( g(v_i, v_j, s, t) \) begins by visiting either the leftmost node \( v_i \) or the rightmost node \( v_j \) of the interval \( [v_i, v_j] \). If $ s = v_i $, then 
\begin{align*}
 g(v_i, v_j, s, t) = \min  \{ f(w + w_i) \cdot d(v_i,v_{i+1}) +\\ g(v_{i+1}, v_j, v_{i+1}, t), 
f(w + w_i) \cdot d(v_i, v_j) + g(v_{i+1}, v_j, v_j, t)
\}.
\end{align*}
 If $ s = v_j $, then 
\begin{align*}
 g(v_i, v_j, s, t) = \min  \{f(w + w_j) \cdot d(v_j,v_{j-1}) +\\ g(v_{i}, v_{j-1}, v_{j-1}, t), 
f(w+ w_j) \cdot d(v_j, v_i) + g(v_i, v_{j-1}, v_i, t)
\}.
\end{align*}

Using dynamic programming, we recursively compute \( g(v_i, v_j, s, t) \) for all $n^2+n$ subproblems and ultimately obtain \( g(v_1, v_n, s, t) \) for \( s \in \{v_1, v_n\} \). The optimal tour is then given by the minimum of the following two values: $g(v_1, v_n, v_1, t) + f(0) \cdot d(t, v_1)$ and $ g(v_1, v_n, v_n, t) + f(0) \cdot d(t, v_n)$.

\begin{theorem} When the metric is a path graph, then the W-TSP problem with a fixed starting node $t$ can be solved optimally in time $O(n^2)$.
\end{theorem}

To find the optimal solution for the W-TSP problem, we consider all possible starting nodes $t$, compute the optimal tour for each, and select the best among them.

\begin{theorem} When the metric is a path graph, then the W-TSP problem with a freely chosen starting node can be solved optimally in time $O(n^3)$.
\end{theorem}

\subsection{Linear speed decrease: An $O(\log(n))$-approximation algorithm}

\begin{figure}[t]
    \centering
    \includegraphics[width=8cm]{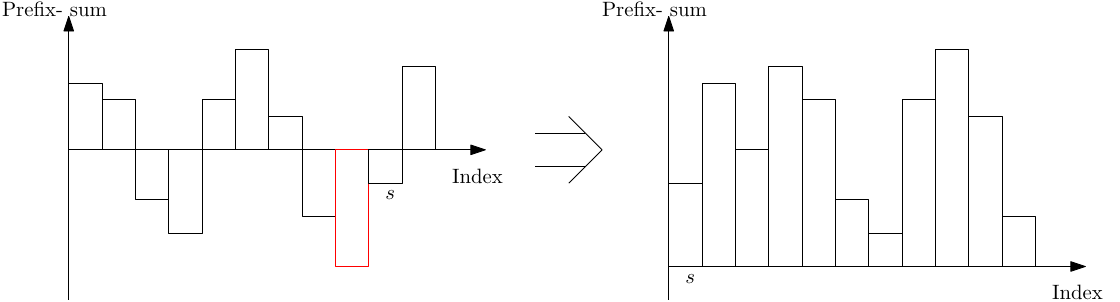}
    \caption{The idea to find a suitable index $s$. One can simply compute the Prefix sums starting with index $1$ (left picture) and find the index with the largest negative value (red bar). By choosing $s$ to be the index right after it one obtains a solution where every prefix sum is positive (right picture).}
    \label{fig:start}
\end{figure}

In the following we consider the case that the cost function $f(w)$ corresponds to a linear speed decrease. To be more precise we are given an initial speed $\nu_{max}$ and another speed value $\nu_{min}$ once we have collected up all items. Then $f(w) = \frac{1}{\nu_{max} - \nu}$, where $\nu = \frac{\nu_{max} - \nu_{min}}{\sum_{i=1}^n w_i}$. For any constant $\delta \in \mathbb{R}_{\geq 0}$ the value $\delta \cdot f(w) = \frac{\delta}{\nu_{max} - \nu w}$ can be interpreted as the time the agent needs if it travels a distance of $\delta$ with speed $\nu_{max} - \nu w$. Our objective becomes minimizing the travel time of the agent while its speed decreases linearly in the weight it has collected up. One may note that this is the standard cost function of the traveling thief problem. 

We will prove that if the first city to visit (i.e. the start node of the tour) can be chosen freely that we can obtain a $O(\log(n))$-approximation for this problem by calculating a constant approximation for the optimal TSP tour (ignoring the weights) and afterwards choosing a feasible start for this tour. This is based on the property that for any given TSP tour $\pi_1,...,\pi_n$ one can always choose a start such that the entire travel time is upper bounded by $O\left(\log(n) \frac{d(\pi)}{\nu_{max}} \right)$.

Scaling the instance, we assume w.l.o.g that $\sum_{i=1}^n w_i = n$, $\sum_{i=1}^n d(\pi_i,\pi_{i+1}) =n$ and $\nu_{max} = n$. Additionally we assume that $\nu_{min} = 0$ as a larger minimum speed will only reduce the gap between the travel time and the length of the TSP tour. Now we define for any $i \leq n$ the value $a_i = d(\pi_i,\pi_{i+1}) - w_{i+1} + \epsilon ( 1 - w_{i +1})$. By using that the average distance between two consecutive cities on the tour is identical to the average weight and that the average weight is $1$ one obtains that $\sum_{i=1}^n a_i = 0$. Thus there exists an $s \leq n$ such that for all $i \leq n$ it holds $\sum_{j = s}^{s + i -1} a_i \geq 0$, where $a_j = a_{j-n}$ for any $j > n$  \cite{lovasz2007combinatorial}. In Figure \ref{fig:start} it is shown how to find $s$.

Now we will prove that choosing $\pi_s$ as a starting point will actually give us the desired property. Intuitively the fact that the sum of the $a_i$ stays positive ensures that the ratio between the length of the tour that still has to be traveled and the weight the agent has not collected up yet (which is equal to the remaining speed) is only constant (this is ensured by the first part of the definition of $a_i$). Additionally it also ensures that the last weights to be collected up are not too tiny (second part of the definition of $a_i$) which would result in the agents speed approaching $0$ at the end. By formalizing and combining these two properties one can upper bound the duration of the tour by an integral that lies within $O(\log(n))$:

\begin{lemma}  If the agent starts at a city $\pi_s$ on the tour $\pi_1,...,\pi_n$, such that $\sum_{j = s}^{s + i -1} a_j \geq 0$ for all $i \leq n$, the entire duration of the tour is upper bounded by $(1 + \epsilon) (\log(n) + 1 -\log(\epsilon))$ for any $\epsilon>0$. 
\end{lemma} 

\begin{proof}
W.l.o.g. we may assume that $s = 1$. Now for any $ i \leq n$ we define $l_i = \sum_{j = 1}^i d(\pi_j,\pi_{j+1})$ and $W_i = \sum_{j = 1}^i w_{j + 1}$. It holds for all $i$ that:
\begin{align*}
 &0 \leq \sum_{j = 1}^i a_j\\
\Rightarrow &0 \leq \sum_{j = 1}^i \left(d(\pi_j,\pi_{j+1}) - w_{j+1} + \epsilon ( 1 - w_{j +1})\right)\\
\Rightarrow &0 \leq l_i - (1 + \epsilon) W_i + \epsilon i\\
\Rightarrow &n - W_i \geq n- \frac{l_i}{1 + \epsilon}-\frac{\epsilon i}{1+\epsilon}\\
\Rightarrow &n - W_i \geq \frac{n+\epsilon n - l_i- \epsilon i}{1 + \epsilon}\\
\Rightarrow &n - W_i \geq \frac{n - l_i}{1 + \epsilon}
\end{align*}

One may note that $n - W_{i}$ is exactly the speed of the agent after it visited city $\pi_{i+1}$. We want to use this to lower bound the velocity, once the distance $l$ is traveled, independently from the exact position of the cities on the tour. Let $i$ be the maximum index still fulfilling that $l_i \leq l$. Then the speed at point $l$ is lower bounded by:
\begin{equation*}
n -W_{i-1} \geq n-W_i \geq \frac{n - l_i}{1 + \epsilon} \geq \frac{n - l}{1 + \epsilon}
\end{equation*}

Using this we can upper bound the duration until the agent traveled distance $n - \epsilon$ by the following integral:
\begin{align*}
\int_{0}^{n- \epsilon} \frac{1}{\frac{n - t}{1 + \epsilon}} dt &= \int_{0}^{n- \epsilon} \frac{1 + \epsilon}{n - t} dt\\
&= (1 + \epsilon) \int_{\epsilon}^n \frac{1}{t} dt\\
&= (1 + \epsilon) (\log(n) -\log(\epsilon))
\end{align*}

All that is left to consider is the duration to cover the last interval on the tour of length $\epsilon$. Since $0 \leq \sum_{j = 1}^{n - 1} a_j$ it holds that $0 \geq a_{n} = d(\pi_n,\pi_1) - w_1 + \epsilon (1-w_1)$ which implies $w_1 \geq \frac{\epsilon}{1 + \epsilon}$. One may note that $n-W_{n-1} = w_1$ is the final velocity of the agent. Thus the duration for the final part of the tour is upper bounded by $\epsilon/w_1 \leq 1 + \epsilon$ and the entire duration by $(1 + \epsilon) (\log(n) + 1 -\log(\epsilon))$.
\end{proof}

One may note that even if the agent would have traveled through the tour with maximum speed $n$ he would have needed duration $1$. Thus one can incorporate a weight function, corresponding to a linear speed loss, into a TSP tour while only losing a factor in $O(\log(n))$. By combining this approach with a constant approximation algorithm for metric TSP we directly obtain an $O(\log(n))$ approximation algorithm for weighted TSP with this kind of cost functions.

\section{Experiments}  
\label{sec5}
In our experiments, we utilize TTP benchmark instances introduced in \cite{DBLP:conf/gecco/PolyakovskiyB0MN14} that are built on the TSP instances from TSPLIB~\cite{DBLP:conf/gecco/PolyakovskiyB0MN14}. 
The instances are publicly available.\footnote{\url{https://cs.adelaide.edu.au/~optlog/TTP2017Comp/}}
Our benchmark set consists of Euclidean TSP benchmark instances with $n$ cities given as coordinates $(x_i, y_i)$, $1\leq i \leq n$ and Euclidean distances between the cities. For the TTP instances, items are placed at the different cities except that starting city.
We use TTP instances based on the TSP instance 
which different numbers of items per city, namely $1, 2, 5, 6, 10, 20$ resulting in TTP instances with $m = 50, 100, 150, 250, 300, 500, 1000$ items with $51$ cities. In addition, we use TTP instances based on the TSP instances which contain $101$ and $575$ cities. We use $1$ item per city resulting in $m = 100, 574$ items, respectively. 
For the experiments with our approach for the path metric, we set all $y_i$-coordinates in the considered TSPLIB instance to $0$.

For comparison, we are using the Algorithm $S5$ which iteratively executes the process defined in the Algorithm $S1$ introduced in~\cite{DBLP:conf/gecco/FaulknerPS015} as a greedy heuristic solver for the classical TTP until a specified time bound is fulfilled. $S1$ applies the Chained Lin-Kernighan heuristic to generate high-quality TSP tours and utilizes the \emph{PackIterative} algorithm to construct packing plans by selecting items based on their score, considering weight, profit, and distance to the last city. 
Note that $S5$ has achieved strong results across various TTP benchmark instances and is recognized as a top-performing algorithm for the problem~\cite{DBLP:conf/gecco/FaulknerPS015}. We run the algorithm on each considered TTP instance for a maximum duration of $10$ minutes on a MacBook Pro 2022 with a Apple M2 chip and 24GB RAM.

We start from the solution \( s = (\pi, x) \) produced by the S5 heuristic, where \( \pi \) denotes the tour and \( x \) is the corresponding packing plan. Fixing the packing plan \( x \), we derive a W-TSP instance in which the cost of a tour is determined by the accumulated weight induced by \( x \) along the tour. We then apply our tour-optimization algorithms to this induced W-TSP instance and evaluate whether, and by how much, they can improve upon the tour \( \pi \) returned by S5 under the same packing plan. 
Since the benchmark TTP instances specify a fixed starting node, all evaluated tours use the same starting node to ensure a fair and consistent comparison.

\subsection{Performance Analysis}

Tables~\ref{tab:summary_statistics} 
summarizes the performance of our proposed dynamic programming (DP) algorithm on instances derived from the TSPLIB dataset. Overall, our algorithm demonstrates consistent and sometimes significant improvements in the route planning objective compared to the benchmark approach. 
In total, we evaluated 25 
problem instances. As shown in Table~\ref{tab:summary_statistics}, the  average improvement in route cost is 13.58\%, with the  maximum improvement reaching 69.71\%, and  no degradation observed in any instance (i.e., the worst case is 0\% improvement, indicating the same with the benchmark).

\begin{table}[t]
    \centering
    \renewcommand{\arraystretch}{1.2} 
    \begin{tabular}{l c}
        \toprule
        \textbf{Metric} & \textbf{Value} \\
        \midrule
        Total instances analyzed & 25 \\
        Average improvement & 13.58\% \\
        Best improvement & 69.71\% \\
        Worst improvement & 0.00\% \\
        \bottomrule
    \end{tabular}
    \caption{
    Summary statistics of the experiment. A minimum improvement of 0.00\% indicates that, in some instances, the benchmark tour is already optimal for the induced W-TSP under the fixed packing plan. Since the DP approach computes an optimal solution for the path-metric W-TSP, its solution quality is never worse than that of the benchmark.” 
    }
\label{tab:summary_statistics}
\end{table}

To gain insight into how our DP algorithm achieves these improvements, we provide two visual comparisons of the paths produced by our method and the benchmark, as shown in Figures~\ref{fig:simulated} and~\ref{fig:scatter}. Regarding Figures~\ref{fig:simulated} and~\ref{fig:scatter}, we clarify that the Y-axis in these figures does not represent the actual geometric \( y \)-coordinates of the TSPLIB instances. For the path-metric experiments, all nodes indeed lie on a line.  
The Y-axis in Figures~\ref{fig:simulated} and~\ref{fig:scatter} is used solely for visualization purposes: it represents an artificial level assigned according to the visiting order in the tour. This layout allows us to visually distinguish the traversal order of nodes along the path and to highlight differences (or similarities) between tours produced by our algorithm and the benchmark. In particular, plotting all nodes on a single horizontal line would result in heavily overlapping edges and labels, making the comparison difficult to interpret. The visualization therefore separates nodes vertically to clearly illustrate the tour structure, while preserving the correct X-coordinates that define the path metric.

Figure~\ref{fig:simulated} shows a case where the path produced by our algorithm is identical to that of the benchmark. Both follow the same route pattern, indicating that the benchmark already performs optimally or near-optimally. This shows that DP algorithm pattern is able to recognize when no further improvement is possible and avoid unnecessary computation.  
In contrast, Figure~\ref{fig:scatter} illustrates a case where our algorithm achieves a 21.6\% improvement. The benchmark path proceeds in one direction for a longer distance while collecting nodes (and their associated weights) before turning back. This strategy can lead to higher costs, as the cost function is non-decreasing with respect to the total collected weight. In comparison, our DP algorithm strategically defers collecting node weights until the return path after reaching an endpoint. This allows it to collect items more efficiently on the way back, resulting in a more cost-effective tour.

These results show that while our DP algorithm matches the benchmark in some cases, it outperforms  others by leveraging an optimal item collection substructure. This insight could be used to enhance solution quality in existing TTP studies. 

\begin{figure}[t]
    \centering
    \begin{subfigure}{0.48\textwidth}
        \includegraphics[width=\textwidth]{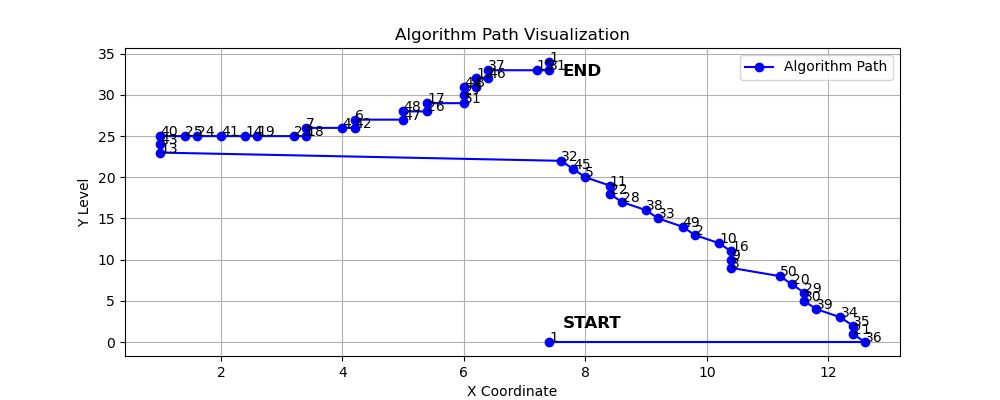}
    \end{subfigure}
    \begin{subfigure}{0.48\textwidth}
        \includegraphics[width=\textwidth]{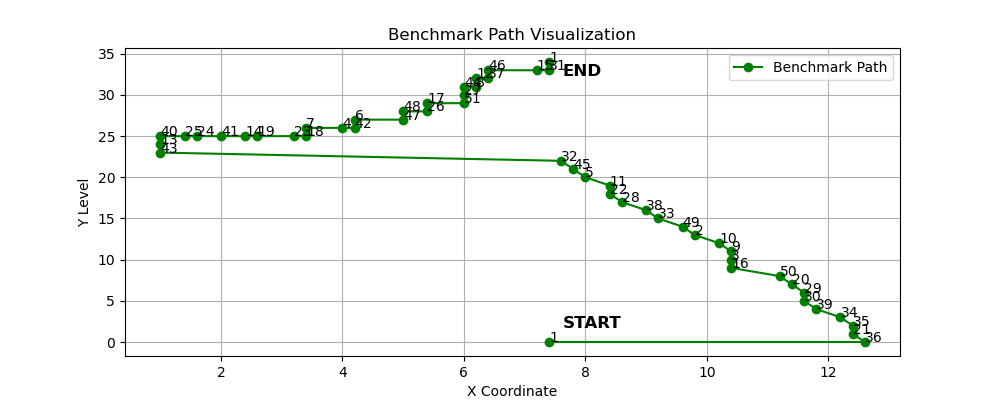}
    \end{subfigure}
    \caption{Same Result as Benchmark}
    \label{fig:simulated}
\end{figure}

\begin{figure}[t]
    \centering
    \begin{subfigure}{0.48\textwidth}
        \includegraphics[width=\textwidth]{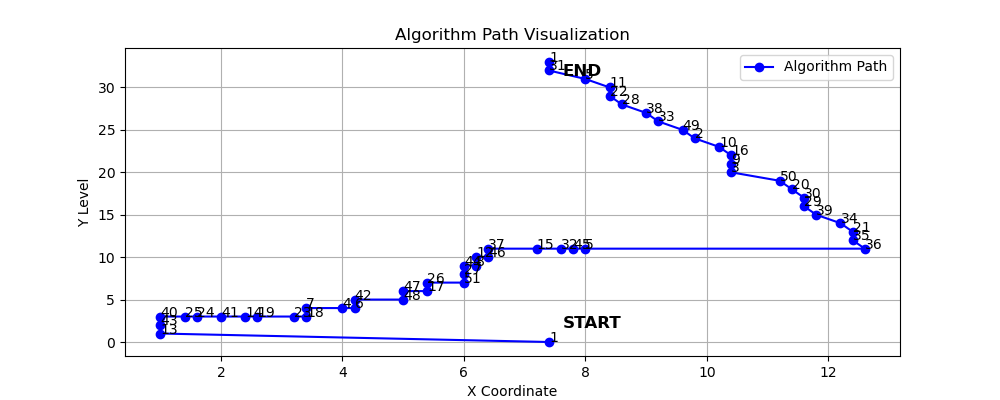}
    \end{subfigure}
    \begin{subfigure}{0.48\textwidth}
        \includegraphics[width=\textwidth]{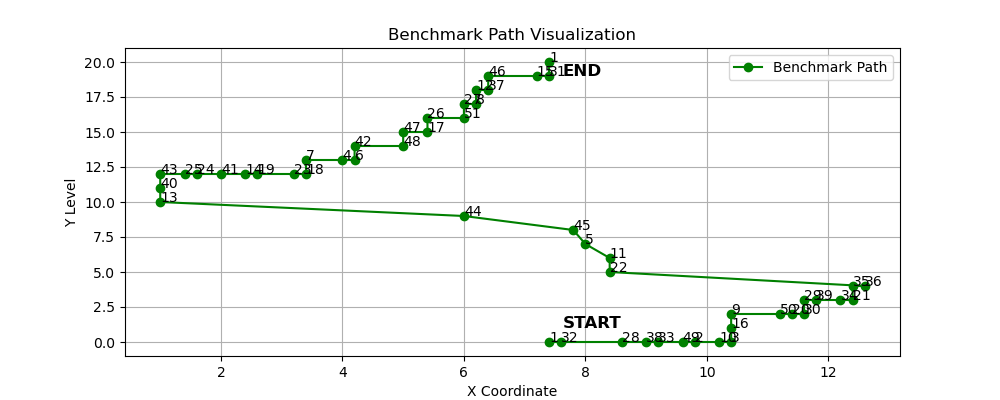}
    \end{subfigure}
    \caption{21.6\% Improvement Compared to Benchmark}
    \label{fig:scatter}
\end{figure}

\subsection{Time Complexity and Empirical Performance}  
We also conducted experiments on a synthetic dataset 
using a \(3 \times 5 \times 5\) configuration, 
combining different node sizes, item sets. 
Specifically, we evaluated $5$ node sizes \((n = 101, 501, 1001, 1501, 2001)\), where one node is designated as the depot and the remaining nodes are randomly distributed. For each number of nodes, we generated $5$ items per city resulting in instances with $m=500, 2500, 5000, 7500, 10000$ items, respectively. 
In the DP implementation, we observed that the step-wise runtime grows quadratically with respect to the number of nodes \(n\), which is consistent with the theoretical time complexity of our DP algorithm.

To improve scalability, we introduce a clustering strategy that aggregates nearby items into \( \sqrt{n} \) clusters based on their spatial locations. Specifically, for each item, we use the coordinates of the node at which it is located and apply a standard \(k\)-means clustering algorithm with \( k = \sqrt{n} \). Each cluster therefore consists of items whose locations are close in the underlying metric.
For each cluster, we construct a single representative (clustered) item as follows. The profit and weight of the representative item are set to the sum of the profits and weights of all items in the cluster. The representative item is assigned to a center node chosen as the node of the item whose location is closest (in Euclidean distance) to the cluster centroid. We maintain an explicit mapping from each representative item to the set of original items it aggregates, which allows the clustered solution to be interpreted in terms of the original instance. 
After clustering, the dynamic programming algorithm is applied to the reduced instance consisting of the clustered items. This significantly reduces the problem size, allowing the DP to run independently on the clustered representation. The resulting solution is then mapped back to the original items using the maintained cluster mapping. Empirically, this approach reduces the overall runtime to be approximately linear in \( n \) (see Figure~\ref{fig:cluster}), while increasing the total cost by only a small margin.

As shown in Table~\ref{tab:statistics_summary}, the preprocessing step including clustering and the computation of segment weights and distances, reduced runtime by an average of \(-82.67\%\), while the DP step achieved an even greater reduction of \(-99.28\%\). The average cost increase across all instances was only \(1.18\%\), indicating that our clustering-based approximation remains highly effective with minimal trade-offs in solution quality.

\begin{figure}[t]
    \centering
    \begin{subfigure}{0.48\textwidth}
        \includegraphics[width=0.9\textwidth]{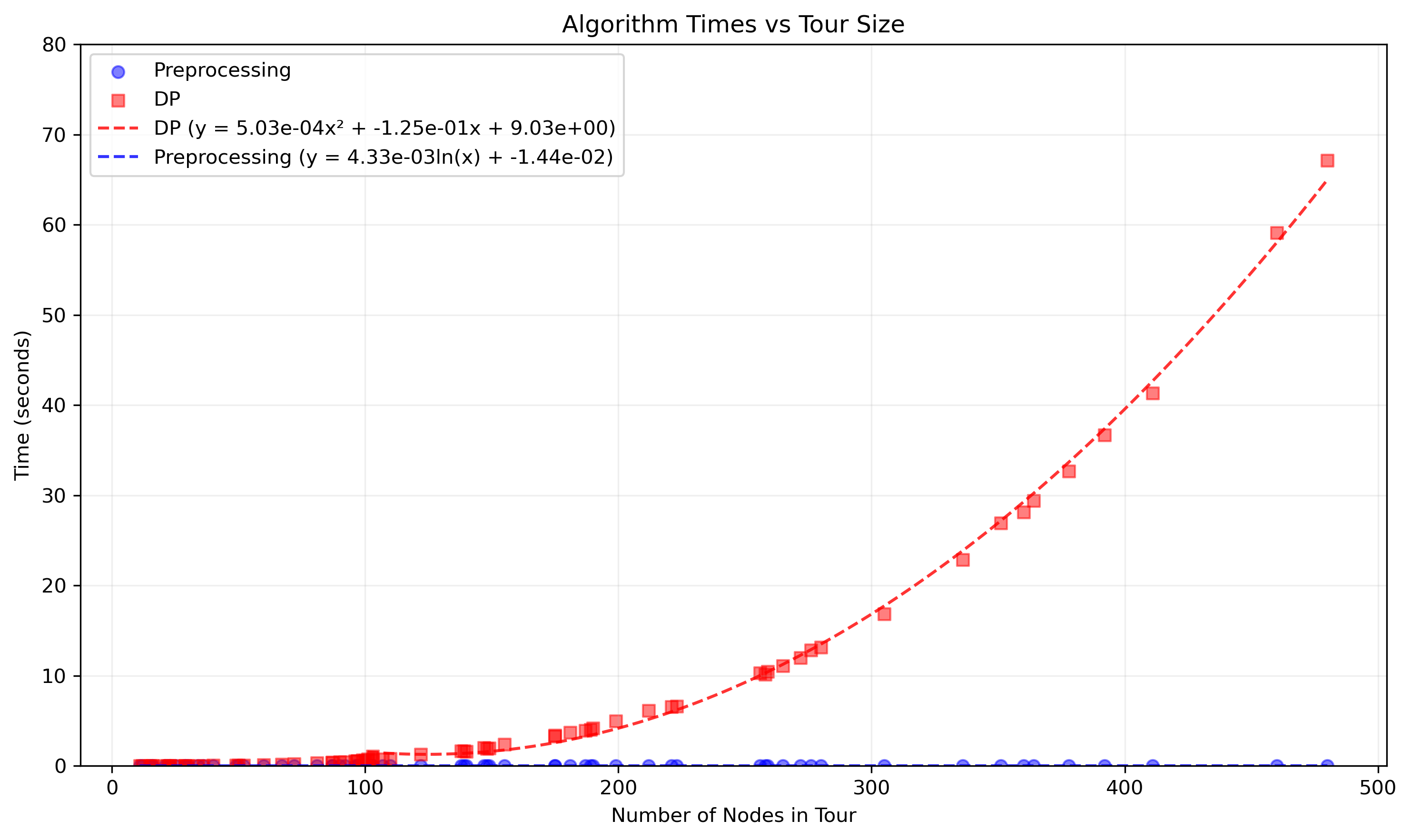} 
    \end{subfigure}
    \begin{subfigure}{0.48\textwidth}
        \includegraphics[width=0.9\textwidth]{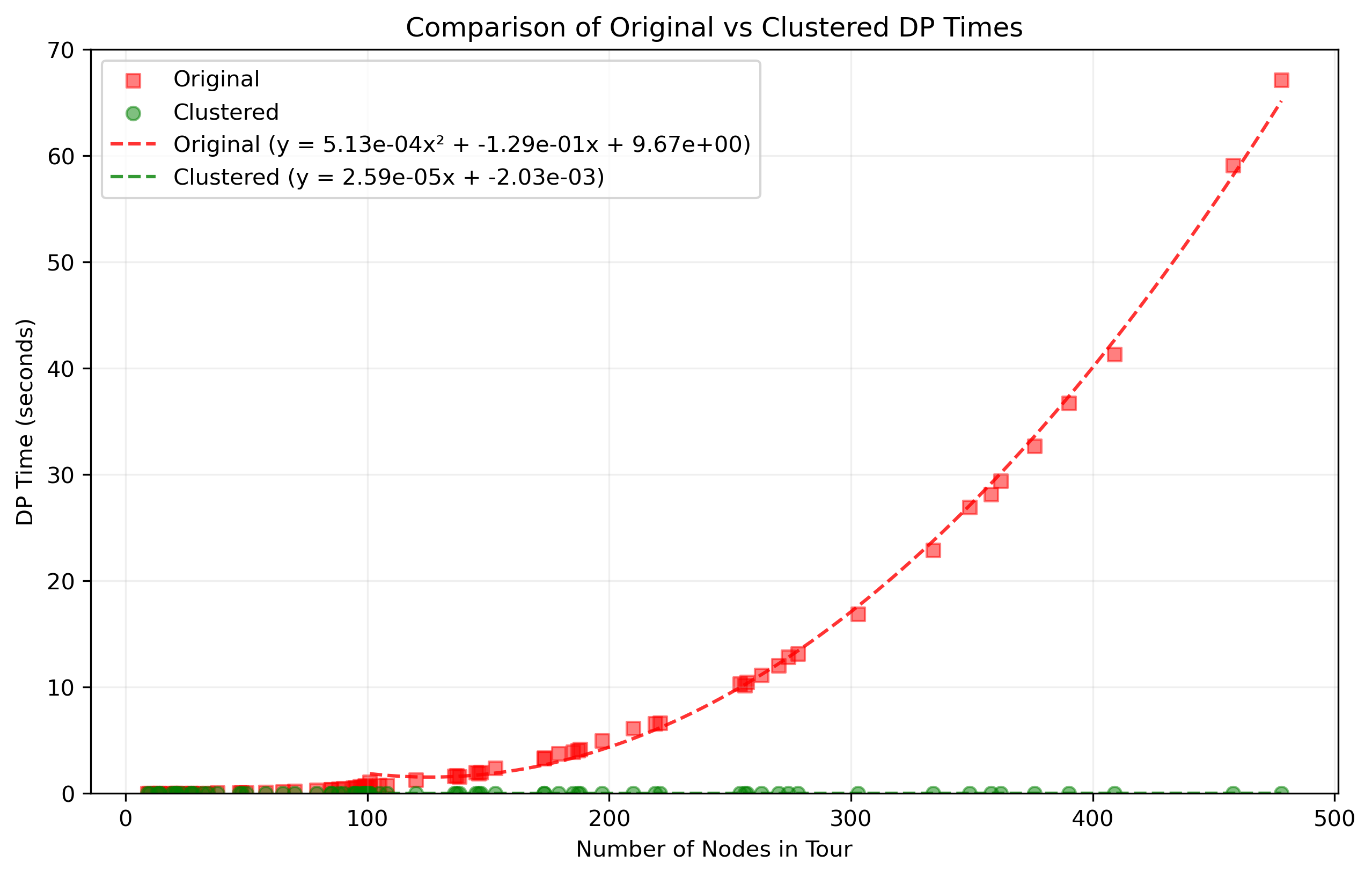} 
    \end{subfigure}
    \caption{Growth of Running Time with Increasing Problem Size}
    \label{fig:cluster}
\end{figure}

\begin{table}[t]
\centering
\begin{tabular}{lrrrr} 
\toprule
\textbf{Metric} & \textbf{Mean} & \textbf{Median} & \textbf{Min} & \textbf{Max} \\
\midrule
Pre Time & -82.67 & -93.73 & -98.51 & -30.14 \\
DP Time            & -99.28 & -99.87 & -99.98 & -91.94 \\
Cost               &  1.18  &  0.45  &   0.00 &   5.04 \\
\bottomrule
\end{tabular}
 \caption{Statistics Summary (\%)}
 \label{tab:statistics_summary}
\end{table}
 \vspace{-1em}

\section{Conclusion}  
Motivated by multi-component problems such as the Traveling Thief Problem, we studied a generalized version of the Traveling Salesman Problem, where the unit cost of traversing an edge depends on the total weight collected from previously visited nodes, instead of $1$. We analyzed this problem under various settings and present algorithms for special cases, including specific metrics (e.g., path metrics, star metrics) and special cost functions (e.g., linear functions representing speed).  Our experimental results demonstrate that, for practically relevant instances, the proposed algorithm—the dynamic programming approach for the path metric—can substantially improve solution quality over existing heuristics while remaining computationally efficient. 

\ignore{
\section*{Acknowledgements} 
This work was supported by the Australian Research Council through grant FT200100536.
}
\bibliographystyle{plainurl}
\bibliography{references}
\end{document}